\documentclass[10pt,conference]{IEEEtran}
\usepackage{amsfonts}
\usepackage{times}
\usepackage{latexsym}
\usepackage{amssymb}
\usepackage{amsmath}
\usepackage{cite}
\usepackage{verbatim}

\newcommand{\bydef}{\triangleq}

\def\bydef{:=}

\def\bb0{{\mathbb{0}}}

\def\bydef{:=}

\def\bb{{\mathbf{b}}}

\def\bh{{\mathbf{h}}}

\def\b0{{\mathbf{0}}}

\def\b1{{\mathbf{1}}}

\def\bbE{{\mathbb{E}}}

\def\bydef{:=}

\def\sf0{{\mathsf{0}}}

\def\SIR{{\mathsf{SIR}}}

\def\nn{\nonumber}

\usepackage{graphicx}
\usepackage{amssymb}
\usepackage{amsfonts}
\usepackage{amsmath}
\usepackage{latexsym}
\usepackage{epstopdf}
\usepackage{cite}
\IEEEoverridecommandlockouts
\begin{document}
 %\IEEEoverridecommandlockouts

\newtheorem{thm}{Theorem}
\newtheorem{lemma}{Lemma}
\newtheorem{rem}{Remark}
\newtheorem{exm}{Example}
\newtheorem{prop}{Proposition}
\newtheorem{defn}{Definition}
\newtheorem{cor}{Corollary}
\def\proof{\noindent\hspace{0em}{\itshape Proof: }}
\def\endproof{\hspace*{\fill}~\QED\par\endtrivlist\unskip}
\def\bh{{\mathbf{h}}}
\def\SIR{{\mathsf{SIR}}}
\def\SINR{{\mathsf{SINR}}}

\title{Finite-Horizon Optimal Transmission Policies for Energy Harvesting Sensors}

\author{\IEEEauthorblockN{Rahul Vaze}
\IEEEauthorblockA{School of Technology and Computer Science\\
Tata Institute of Fundamental Research,
Mumbai 400005, India\\
Email:  vaze@tcs.tifr.res.in}
\and
\IEEEauthorblockN{Krishna Jagannathan}
\IEEEauthorblockA{Department of Electrical Engineering\\
IIT Madras,
Chennai 600 036, India\\
Email:  krishnaj@ee.iitm.ac.in} \thanks{Resesarch support from ITRA Grant 13X200 is gratefully acknowledged.}
}
\maketitle
\begin{abstract}
In this paper, we derive optimal transmission policies for energy harvesting sensors to maximize the utility obtained over a finite horizon. First, we consider a single energy harvesting sensor, with discrete energy arrival process, and a discrete energy 
consumption policy. Under this model, we show that the optimal finite horizon policy is a threshold policy, and explicitly characterize the thresholds, and the thresholds can be precomputed using a recursion. Next, we address the case of multiple sensors, with only one of them allowed to transmit at any given time to avoid interference, 
and derive an explicit optimal policy for this scenario as well.
\end{abstract}

\section{Introduction}
Energy harvesting is a paradigm where wireless sensor nodes have the ability to recharge their batteries from their surrounding environment, by using solar, heat or vibration energy. Due to the potential for energy harvesting nodes to significantly enhance the lifetime of a sensor network, there has been considerable interest in this paradigm; see~\cite{paradiso2005energy,sudevalayam2011energy} for surveys and examples.

Since the amount of ambient energy available for the nodes to replenish their batteries can vary unpredictably, it is important for energy harvesting nodes to make judicious use of available energy. In particular, energy harvesting nodes often have to tradeoff between transmitting at a particular time to obtain a certain utility, and saving energy for potential future use. Clearly, the characteristics of the renewable energy source has a key role to play in this tradeoff. In addition, the instantaneous utility obtained by expending a given amount of energy could also be time-varying, due to inherent importance of the data being sent, or simply due to the channel fading. We address this basic tradeoff in this paper. Specifically, we ask for an optimal energy utilization policy that maximizes the total utility obtained over a finite horizon, when the instantaneous utility is time varying, and the battery replenishes according to a random process.

\emph{Related Work:} In \cite{fu2006optimal}, the authors consider the problem of maximizing the finite horizon throughput of a transmitter sending data over a time-varying channel under a total energy constraint (for example, a fixed battery). 
Assuming that the channel state is revealed to the transmitter before each transmission attempt, the authors develop a dynamic programming
algorithm that provides an optimal policy for the case where the throughput obtained is concave in the energy spent. For the special case when the throughput obtained in linear in the energy spent, the authors derive a closed-form optimal policy.
%The problem considered in this paper is quite similar, although there is an additional complication of random energy arrivals that replenish the battery during the course of the transmission.

For the energy harvesting case, maximizing a time-average utility function over an infinite horizon is considered in \cite{michelusi2012optimal}. Under a Bernoulli energy arrival, and binary energy expenditure model, the authors show that the optimal policy is of the threshold form, with the threshold values being monotonically decreasing in the energy available. However, an explicit characterisation of the thresholds was not possible. In a closely related paper \cite{sinha2012optimal}, the authors derive structural properties such as monotonicity for an infinite horizon  discounted reward Markov decision process. Similar properties are established for the finite horizon case in \cite{ZhangEH2012}.
Another recent paper \cite{wang2013simplicity} proposes computationally simple control policies based on heuristics that achieve near-optimal performance in the finite-horizon case with a finite battery. Finally, \cite{huang2011utility,sharma2010optimal} take a queue stability view of energy harvesting networks, using Lyapunov optimization techniques.

\emph{Contributions:} In the present paper, we derive optimal transmission policies for energy harvesting nodes to maximize a utility function over a finite horizon. The utility function is assumed to be a monotonic increasing function of the energy used by the node. First we consider a single node case, where the node is allowed to transmit any discrete quantum of available energy, under a general energy burst arrival model. We show that the optimal transmission policy is of the threshold form, and we also explicitly characterize the threshold values, where the thresholds can be computed using a recursion. Next, we show that for a case with more than one energy harvesting node, the same results can be derived, although the structure becomes more cumbersome.

%where the objective is to maximize a utility function over a finite horizon.  We first consider a single energy harvesting node with a finite battery, which either transmits one unit of energy during each time slot, or does not transmit at all. The objective is to maximize a utility function over a finite horizon. Under Bernoulli energy arrivals, we show that the optimal transmission policy is of the threshold form. We also explicitly characterize the threshold values using a non-linear recursion. Next, we generalise to a scenario where the node is allowed to transmit using any discrete quantum of available energy. We assume that the utility obtained is a monotonic increasing function of the energy used by the node. Under a general energy burst arrival model, we derive the optimal policy to maximize the total utility obtained over a finite horizon. The optimal policy during each time slot has the structure of an `$m$-ary threshold policy,' where the threshold values can be recursively precomputed, and $m$ is the amount of energy available during that time slot. Finally, we show that for a case with more than one energy harvesting node, the same framework can be extended, although the structure becomes more cumbersome.
%
%\emph{Brief discussion on numerical results here}

%This paper is organised as follows.
\section{System Model}
We consider slotted time, and a single node that harvests energy from the environment. Let $E_k$ be the amount of energy  harvested at time slot $k.$ The node is assumed to have a finite battery of size $B$ in which it stores the harvested energy. For each time $k$, a realization of channel $h_k\ge 0,$ between the transmitter and the receiver is revealed to the transmitter, after which the node transmits using energy $F_{k}\ge 0,$ and obtains a payoff of $r_k(h_k,F_k).$ Although we can allow for any monotonically increasing payoff function, we take $r_k=\log(1+F_{k}h_{k})$. We assume that $h_{k}$ is drawn i.i.d. from some distribution $\phi$.
Let the energy available at time $k$ be denoted by $U_{k}.$ Thus,  $U_{k+1} = \min\{U_{k} + E_{k+1} - F_{k}, B\}$.

We consider the finite-horizon problem, where the expected total payoff over $n$ time slots is
$P = \bbE\{\sum_{i=1}^n \log(1+F_i h_i)\}.$ We are interested in finding the energy utilization $F_i$ that maximizes the total payoff $P,$ under the energy neutrality constraint $\sum_{i=1}^k F_i \le \sum_{i=1}^k E_i$ for every $k=1,\ldots,N$.

Under a general energy arrival and energy consumption model, the problem is difficult to solve. To see this, consider the simple case of $n=2$ time slots. At time slot $1$, given energy $E_1 \le B$, the decision is to determine the energy $F_{1}$ to transmit. Rewriting $P$ as a dynamic program, with payoff at time slot $1$ as $P_1(F_1, E_1)= 
\max_{F_1 \le E_1} r_1(h_1,F_1) +  \bbE\{r_2(h_2, (\min\{E_1 - F_1 + E_2, B\}))\} $,
%\
%\begin{eqnarray*}
%P_1(F_1, E_1)&=& \max_{F_1 \le E_1} \log(1+F_1h_1) +  \\
%&&\bbE\{\log(1+h_2(\min\{E_1 - F_1 + E_2, B\}))\},
%&=& \max_{F_1 \le E_1} F_1h_1 +  \\
%&&\bbE\{h_2\}\bbE\{\min\{E_1 - F_1 + E_2, B\}\}, \\
%\end{eqnarray*}
where the second expectation is over the random energy arrival $E_2$. Since the expectation in the second term  depends on the carry-forward energy $E_1 - F_1$, there is no easy way to establish the optimal value of $F_1$ for general energy arrivals. Indeed, even without random energy arrivals, it is difficult to explicitly characterize the optimal policy; however, as shown in \cite[Section II E]{fu2006optimal} the problem is concave, and lends itself to fast numerical solutions. In this paper, we overcome this difficulty by restricting ourselves to discrete energy consumption and discrete burst energy arrivals. With this restriction, we shall see that an explicit characterisation becomes feasible.

%For the special case when $\bbE\{\min\{E_1 - F_1 + E_2, B\}\} = \min\{E_1 - F_1 + \mu, B\}$, then we can derive the optimal policy for the finite-horizon case.

%
%A restricted case of this problem with a Bernoulli energy arrival ($E_t=\{0,1\}$) and binary energy
%consumption $F_{t} =\{0,1\}$ with a linear payoff of $F_th_t$ under the infinite-horizon setup has been considered in \cite{}, where it has been shown that  the optimal policy is a threshold policy, however, without deriving the actual optimal threshold. The binary energy consumption model is equivalent to deciding whether the policy should utilize a typical slot or not.
%In this work, we derive the optimal policy for the Bernoulli energy arrival and binary energy consumption  constraint under the finite-horizon setup.

%The source is only allowed to transmit a unit amount of energy in any time slot if at all. If source transmits a unit amount of energy at time $t$, it gets payoff $h_t$, otherwise the payoff is zero. Let $F_i=1$ if the source transmitted a unit amount of energy at time $i$, otherwise $F_i=0$.

%Let the
%For the moment assume that the source has unbounded size battery. Let $E_i=1$ if the source receives a unit amount of energy at time $i$, otherwise $E_i=0$.
%The source is only allowed to transmit a unit amount of energy in any time slot if at all. If source transmits a unit amount of energy at time $t$, it gets payoff $h_t$, otherwise the payoff is zero. Let $F_i=1$ if the source transmitted a unit amount of energy at time $i$, otherwise $F_i=0$.

\section{Optimal Policy for a Single Node}\label{sec:singlesource}
\subsection{Binary Energy Arrival and Consumption Model}
We first consider the case of Bernoulli energy arrivals and binary energy consumption model, i.e. $E_k, F_k\in \{0,1\}$. Let $E_k=1$ with probability $p$. For this case, we  derive a finite horizon optimal policy. With abuse of notation, for this subsection, 
we let $r_k = r_k(h_k,1)$, since $F_k \in \{0,1\}$. \begin{thm}
Under Bernoulli energy arrivals and binary energy consumption model, the optimal transmission policy is given by the following threshold rule. 
$$F_{k}^{\star} =\left\{\begin{array}{cc}1&\quad {\rm if} \quad r_k + \gamma_{k+1}^{m-1} > \gamma_{k+1}^{m},\\   0&\quad \rm otherwise,\end{array}\right.$$ where $m$ is the energy available at time slot $k.$ The thresholds are given by the following recursion:   $\gamma_n^0 = p\bbE\{r_n\}$ and $\gamma_n^i = \bbE\{r_n\}, \ i>0$, and
for $0<m\le n-k$
\begin{eqnarray*}
\gamma_{k}^m
&=& (1-p) \bbE\left\{\max\left\{ r_k + \gamma_{k+1}^{m-1}, \gamma_{k+1}^{m} \right\} \right\} \\
&& + p \bbE\left\{ \max\left\{ r_k + \gamma_{k+1}^{m}, \gamma_{k+1}^{m+1} \right\}  \right\}.
\end{eqnarray*}
For $m\ge n-k+1$,
$\gamma_{k}^{m} = \bbE\left\{ r_k\right\} + \gamma_{k+1}^{n-k}$, and for $m=0$,
$\gamma_{k}^0 = p \bbE\left\{\max\left\{ r_k + \gamma_{k+1}^{0}, \gamma_{k+1}^{1} \right\}\right\}  + (1-p )\bbE\left\{\gamma_{k+1}^{0}\right\}$.
\end{thm}

\begin{proof}
Let $c_{k-1}, k=1,\dots,n,$ be the carried over energy from time slot $k-1$ to time slot $k$, and
$U_k = \min\{c_{k-1}+E_k, B\}$ be the total energy available at time slot $k,$ with
$c_0=0$, and $c_k = U_k -F_{k}$.

Then the optimization problem can be posed in the dynamic programming format, by writing the pay-off at time slot $k=1,\dots,n$ as $P_k(c_{k-1}, h_k)
 = \max_{F_k\in \{0,1\}, F_k \le U_k} \left[  \log(1+F_kh_k) + {\bar P}_{k+1}(c_k)\right],$
where ${\bar P}_{k+1}(x) = \bbE\{P_{k+1}(x, h_{k+1})\}$.

To find the optimal transmission policy $\{F_k\}_{k=1}^n$, lets start from the last time slot $n$, where we have ${\bar P}_{n}(c_{n-1}) =  \bbE\{\log(1+h_n\b1_{U_{n}\geq 1})\}$, where   $\b1_{U_n\geq 1}$ is the indicator that $U_n$ is at least 1. Note that ${\bar P}_{n}(x)= \bbE\{r_n\}$ for $x>0$ and  ${\bar P}_{n}(0)= p\bbE\{r_n\}$. Let $\gamma_n^0 ={\bar P}_{n}(0)$ and $\gamma_n^x = {\bar P}_{n}(x)$ for $x>0$.

Next, consider the $n-k^{th}$ time slot (i.e., there are $k$ remaining time slots), and let  $\gamma_{n-k}^m \bydef {\bar P}_{n-k}(m)$.
Then, we have $P_{n-k}(c_{n-k-1}, h_{n-k}) $
\begin{eqnarray}\nn
&=&\max_{F_{n-k}\in \{0,1\}, F_{n-k} \le U_{n-k}} \left[ \log(1+F_{n-k} h_{n-k}) \right. \nonumber\\
&&
 +  \left.{\bar P}_{n-k+1}(U_{n-k} - F_{n-k})\right],\nonumber\\
&\stackrel{(a)}{=}&  \b1_{U_{n-k}=1}\left[\max\left\{ r_{n-k} + \gamma_{n-k+1}^0, \gamma_{n-k+1}^1 \right\} \right] \nonumber\\
&& + \  \b1_{U_{n-k}=2}\left[\max\left\{ r_{n-k} + \gamma_{n-k+1}^1, \gamma_{n-k+1}^2 \right\} \right]\nonumber \\ \nn
&&\dots \\ \nn
&&  + \ \b1_{U_{n-k}=k}\left[\max\left\{ r_{n-k} + \gamma_{n-k+1}^{k-1}, \gamma_{n-k+1}^{k} \right\} \right] \nonumber\\
%&& \left. \ \ - \max\left\{ r_{n-k} + \gamma_{n-k+1}^{k-1}, \gamma_{n-k+1}^k \right\}\right]\\
&&  +\ \b1_{U_{n-k}\geq k+1}\left[r_{n-k} + \gamma_{n-k+1}^k \right]\nonumber \\
%&& \left. \ \ - \max\left\{ r_{n-k} + \gamma_{n-k+1}^k, \gamma_{n-k+1}^{k+1} \right\} \right]  \\
%&& \ + \b1_{\min\{c_{n-k-1}+E_{n-k}-k-1} \min\left\{r_{n-k}, \gamma_{n-k+1}^{k+1} -  \gamma_{n-k+1}^k\right\}\\
&& + \b1_{U_{n-k}=0}\gamma_{n-k+1}^0.\label{eq:n-k}
%
%\max_{a_{n-k}\in \{0,1\}, a_{n-k} \le c_{n-k-1} + E_{n-k}} \left[ a_{n-k} r_{n-k} +  \bbE\{r_{n-k}\}\bbE\{\b1_{c_{n-k-1} + E_{n-k}}\} \right],\\
%&=& \b1_{U_{n-1}}\max\left\{ r_{n-1} + \bbE\{h\}p, \bbE\{h\} \right\} + \b1_{U_{n-1}-1}\min\left\{ r_{n-1}, (1-p)\bbE\{h\} \right\} \\
%&& + \b1^c_{U_{n-1}} p\bbE\{h\},
\end{eqnarray}
To parse $(a),$ note that if at time slot $n-k$, the available energy $U_{n-k}$ is at least $k+1$, then $F_{n-k}=1$, since there are only $k$ more slots left and hence the node should transmit in time slot $n-k$ to get the payoff of $r_{n-k}+\gamma_{n-k+1}^k$. This explains the second last term. Similarly, if $c_{n-k-1}+E_{n-k} =m$, where $m <k+1$, then the choice is between using a unit energy in time slot $n-k$ or using all the $m$ units of energy in future time slots. %In the above expression, with $c_{n-k-1}+E_{n-k} =m$, only the first $m$ terms are non-zero with telescopic sums.

In order to obtain the optimal solution $F_{n-k}$ when there are $m$ units of energy (i.e, $U_{n-k}=m$), $1\leq m\leq k$, we compare the two arguments inside the maximum of the $m$th term. In particular, we get that the optimal solution is $F_{n-k}=1$ if $r_{n-k}+\gamma_{n-k+1}^{m-1}> \gamma_{n-k+1}^{m} $, and $F_{n-k}=0$ otherwise.
In case $m=0$,
then payoff is $\gamma_{n-k+1}^0$ which is the last term of the expression.

The above expression for $P_{n-k}(c_{n-k-1}, h_{n-k}) $ also suggests a straightforward recursion for computing the thresholds $\gamma_{n-k}^j= \bbE\{P_{n-k}(j, h_{n-k})\}$. Indeed, observe that when $c_{n-k-1}=j,\ 0<j\leq k$ then $U_{n-k}=j+1$ with probability $p$ and $U_{n-k}=j$ with probability $1-p$. Thus, taking the expectation of (\ref{eq:n-k}), we get for $0<j\le k$
\begin{eqnarray*}
\gamma_{n-k}^j &=& \bbE\left\{\max\left\{ r_{n-k} + \gamma_{n-k+1}^{j-1}, \gamma_{n-k+1}^{j} \right\}\right\},\\ &=& (1-p) \bbE\left\{\max\left\{ r_{n-k} + \gamma_{n-k+1}^{j-1}, \gamma_{n-k+1}^{j} \right\} \right\} \\&& + p \bbE\left\{ \max\left\{ r_{n-k} + \gamma_{n-k+1}^{j}, \gamma_{n-k+1}^{j+1} \right\}  \right\}.
\end{eqnarray*}
%Similarly, for $j=0$, we have
%\begin{eqnarray*}
%\gamma_{n-k}^0 &=& p \bbE\left\{\max\left\{ r_{n-k} + \gamma_{n-k+1}^{0}, \gamma_{n-k+1}^{1} \right\}\right\} \\ && + (1-p )\bbE\left\{\gamma_{n-k+1}^{0}\right\}.\end{eqnarray*}
%Finally, for $j\ge k+1$, we have
%\begin{eqnarray*}
%\gamma_{n-k}^{j} &=& \bbE\left\{\max\left\{ r_{n-k} + \gamma_{n-k+1}^{k}\right\}\right\}.\end{eqnarray*}
Similarly, we can obtain $\gamma_{n-k}^j$ for $j=0$ and $j\ge k+1.$
\end{proof}

\subsection{Optimal policy for the general discrete case}
In this section, we consider a more general scenario, where both the energy arrival and transmitted energy can take any discrete value between $0$ and $B$. Here too, we are able to explicitly characterize the optimal finite horizon throughput maximizing policy. Assume that $i$ units of energy arrive during each slot with probability $p_i,\ i=0,1,\dots,B,$ and that this process is i.i.d. across time.
\begin{thm}
Suppose $m$ units of energy are available in slot $k,$ i.e., $U_k=m.$ Then the optimal policy is to transmit  $F^{\star}_{k}=q$ units of energy, where
$q = \arg\max_{j\in\{0,1,\dots,m\}}\log(1+jh_k)+\gamma_{k+1}^{m-j}$.
%and the maximum value is defined as
\end{thm}

\begin{proof}
Similar to the previous section,
%Let $c_k, k=1,\dots,n,$ be the carried over energy from time slot $k$ to time slot $k+1$, with
%$U_k = \min\{c_{k-1}+E_k, B\}$ be the total energy available at time slot $k$ and
%$c_0=0$, and $c_k = U_k -F(k)$.
 the optimization problem can be posed in the dynamic programming format, by writing  the payoff at time slot $k=1,\dots,n$ as 
$$
P_k(c_{k-1}, r_k) = \max_{F_k\le U_k} \left[ \log(1+F_k h_k) + {\bar P}_{k+1}(c_k)\right],$$
where ${\bar P}_{k+1}(x) = \bbE\{P_{k+1}(x, h_{k+1})\}$.

To find the optimal transmission protocol $\{F_i\}_{i=1}^n$, lets start from the last time slot $n$, where we have ${\bar P}_{n}(c_{n-1}) =  \bbE\{\log(1+U_nh_n)\}$.
Let $\gamma_n^m ={\bar P}_{n}(m) = \bbE\{\log(1+h_n\min\{m+E_n,B\})\}$, $m=0,1,\dots,B$. Let $\gamma_{k}^m \bydef {\bar P}_{k}(m)$ and define \begin{equation}\alpha_{k+1}^m\bydef \max_{j\in\{0,1,\dots,m\}}\log(1+jh_k)+\gamma_{k+1}^{m-j}.\label{defn:alpha}
\end{equation}

For the $n-k^{th}$ time slot, we have $P_{n-k}(c_{n-k-1}, r_{n-k}) $
\begin{multline}
=  \max_{F_{n-k}} \left[\log(1+F_{n-k} h_{n-k}) + {\bar P}_{n-k+1}(c_{n-k+1})\right] \\
=  \sum_{m=0}^{B} \b1_{U_{n-k}=m}\alpha_{n-k+1}^m. \label{eqn:telescopic}
\end{multline}

%where $(a)$ is the telescopic decomposition of $P_{n-k}(c_{n-k-1}, r_{n-k}) $ following definition of $\alpha$ from (\ref{defn:alpha}).
%s from the definition of  fact that if at time slot $n-k$, the available energy $U_{n-k} = c_{n-k-1}+E_{n-k} \ge k+1$, then $F_{n-k}=1$, since there are only $k$ more slots left and hence the source should transmit in time slot $n-k$ to get the payoff of $r_{n-k}+\gamma_{n-k+1}^k$. This explains the last but one term. Similarly, if $c_{n-k-1}+E_{n-k} =m$, where $m <k+1$, then the choice is between using a unit energy in time slot $n-k$ or using all the $m$ units of energy in future time slots. In the above expression, with
%$c_{n-k-1}+E_{n-k} =m$, only the first $m$ terms are non-zero with telescopic sums.
In order to obtain the optimal solution $F_{n-k}$, we need to look at the $m^{th}$ term of \eqref{eqn:telescopic}, where $U_{n-k}=m.$ Then, the optimal $F^{\star}_{n-k}=q$, where $q$ is the index $j$ that achieves the maximum in (\ref{defn:alpha}). To find the value of $q,$ we need to know $\gamma_{n-k+1}^j, \ j=0,\dots,m,$ which can be found as follows.

%Using this expression for $P_{n-k}(c_{n-k-1}, r_{n-k}) $ we can also compute $\alpha_{n-k+1}^{m}$ for any $k, m$.
Setting $c_{n-k-1}=j$  and taking the expectation of (\ref{eqn:telescopic}), we get for $0\le j\le B$
\begin{equation}
\gamma_{n-k}^j = \sum_{i=0}^{B-j-1} p_i \bbE\left\{\alpha_{n-k+1}^{i+j} \right\}+\left(\sum_{i=B-j}^B p_i\right)\bbE\left\{\alpha_{n-k+1}^{B} \right\},
\end{equation}
where $p_i$ is the probability that $E_{n-k} = i$.
%For $j\ge k+1$, we have
%\begin{eqnarray*}
%\gamma_{n-k}^{j} &=& \bbE\left\{\max\left\{ r_{n-k} + \gamma_{n-k+1}^{k}\right\}\right\}.\end{eqnarray*}

%For $j=0$, we have
%\begin{eqnarray*}
%\gamma_{n-k}^0 &=& p \bbE\left\{\max\left\{ r_{n-k} + \gamma_{n-k+1}^{0}, \gamma_{n-k+1}^{1} \right\}\right\} \\ && + (1-p )\bbE\left\{\gamma_{n-k+1}^{0}\right\}.\end{eqnarray*}
\end{proof}
\begin{rem} The thresholds $\gamma_{n-k}^{j}$'s depend only on the distribution of 
energy arrivals and channel gains and can be computed ahead of time.
\end{rem}
\begin{figure*}
\begin{equation}\label{eq:telescopic2}
\beta_{n-k}(m_1, m_2) = \max\{r^1_{n-k} + \gamma_{n-k+1}(m_1-1, m_2), r^2_{n-k} + \gamma_{n-k+1}(m_1, m_2-1),  \gamma_{n-k+1}(m_1, m_2)\}.
\end{equation}
\end{figure*}
\section{Multiple Nodes}\label{sec:multiplesources}
In this section, we extend the problem and consider two energy harvesting nodes.
For ease of notation, we consider the Bernoulli energy arrival model, where node $i$ harvests either one unit of energy, or does not harvest any energy. In particular, let $E^i_t$ be the energy harvested by node $i$ in time slot $t.$ We assume $E_i^t=1$ with probability $p_i,$ independently of each other, and independent from slot-to-slot.
At each time slot a new channel realization is given to the two nodes, and to avoid interference, at most one of two nodes is allowed to transmit. We also restrict ourselves to the unit energy consumption model.
If node $i$ transmits at time $t$, then it gets a payoff $\log(1+h^i_t)$. We consider the finite horizon problem
$$P  = \max\bbE\left\{\sum_{t=1}^n \log(1+F^1_t h^1_t) + \log(1+F^2_th^2_t)\right\},$$ such that $F^1_t + F^2_t \le1$,
where $F^i_t \in \{0,1\}$, and $F^i_t =1$ if an unit amount of energy is transmitted from node $i$ at time $t$. The objective is to maximize $P$, under the per-node energy neutrality constraint $\sum_{t=1}^k F^i_t \le \sum_{t=1}^k E^i_t$. Our next result derives an optimal policy for this case.

%We next show that similar to the single node case, even for two nodes, there is an explicit optimal policy to maximize $P$.

\begin{thm}
Suppose $m_i$ units of energy be available in slot $k$ at node $i$. Then the optimal policy is to transmit  $F^i_{k}=1, i=1,2$ if $q=i, i=1,2$, where $q$ is the index of the maximum in  (\ref{eq:telescopic2}). Otherwise, if $q=3$, then
$F^i_{k}=0, i=1,2$.
\end{thm}

\begin{proof}
Let $c^i_k, k=1,\dots,n,$ be the energy carried over by node $i$ from time-slot $k$ to  $k+1.$ Let  $U^i_k = \min\{c^i_{k-1}+E^i_k, B\}$ be the total energy available with node $i$ at time slot $k$ with
$c^i_0=0$, and $c^i_k = U^i_k -F^i_k$.

Then the optimization problem can be posed in the dynamic programming format, by writing  the payoff at time slot $k=1,\dots,n$ as $P_k(c^1_{k-1}, c^2_{k-1} ,h^1_k,h^2_k )$
\begin{eqnarray*}
&=& \max_{F^1_k, F^2_k \in \{0,1\}, F^1_k + F^2_k\le 1, F^i_k \le U_k^i} \left[ \log(1+F^1_t h^1_t) +  \right.\\
 &&\left. \log(1+ F^2_th^2_t) + {\bar P}_{k+1}(c^1_k, c^2_k),\right]
\end{eqnarray*}
where ${\bar P}_{k+1}(x, y) = \bbE\{P_{k+1}(x, y, h^1_{k+1}, h^2_{k+1})\}$.

To find the optimal transmission protocol $\{F^i_k\}_{k=1}^n$, lets start from the last time slot $n$, where we have ${\bar P}_{n}(c^1_{n-1},c^2_{n-1}) =  \bbE\{\max\{r^1_n\b1_{U^1_{n}\ge 0},r^2_n\b1_{U^2_n\ge 0}\}\}$, where $r^i_{n-k} = \log(1+h^i_{n-k})$. Thus, the optimal decision is $F^1_n=1, F^2_n=0$ whenever $h^1_n\b1_{U^1_{n}\ge 0}>h^2_n\b1_{U^2_n\ge 0},$ and vice-versa.
%, where the second expectation in only over energy arrival $E_n$, and $\b1_x=1$ for $x>0$ and $0$ otherwise, and function $\b1^c$ is the complement of $\b1$. Note that ${\bar P}_{n}(x)= \bbE\{r_n\}$ for $x>0$ and  ${\bar P}_{n}(0)= p\bbE\{r_n\}$.
%Let $\gamma_n(x,y) ={\bar P}_{n}(x,y)$, where $x,y\in \{0,1\}$ and $x,y =1$ represent that non-zero energy is .

Let $\gamma_{n-k}(x,y) \bydef {\bar P}_{n-k}(x,y)$.
Consider the $n-k^{th}$ time slot, where there are $k$ more remaining time slots.
Then, we can write $P_{n-k}(c^1_{n-k-1}, c^2_{n-k-1}, h^1_{n-k}, h^2_{n-k}) $
\begin{eqnarray}
 &=&  \sum_{m_1=1}^{B} \sum_{m_1=1}^{B}  \b1_{U^1_{n-k} = m_1} \b1_{U^2_{n-k} = m_2}
 \left[ \beta_{n-k}(m_1, m_2) \right],\label{eq:payoffmultisource}
 \end{eqnarray}%\\\nn
% && \left.-  \b1_{U^1_{n-k} > U^2_{n-k} } \beta_{n-k}(m_1-1, m_2) \right.\\\nn
% &&\left. -  \b1_{U^1_{n-k}< U^2_{n-k}} \beta_{n-k}(m_1, m_2-1)  \right. \\\nn
% && \left.-  \b1_{U^1_{n-k} = U^2_{n-k} } ( \beta_{n-k}(m_1-1, m_2) +  \beta_{n-k}(m_1, m_2-1) \right]\\\nn
%&& +  \b1^c_{U^1_{n-k}} \b1_{U^2_{n-k}}
%\beta_{n-k}(0, 1) \\ \nn
%&& +  \b1_{U^1_{n-k}} \b1^c_{U^2_{n-k}}
%\beta_{n-k}(1,0 ) \\ \label{eq:payoffmultisource} &&+  \b1^c_{U^1_{n-k}} \b1^c_{U^2_{n-k}}
% \beta_{n-k}(0, 0),
%\end{eqnarray}
where $\beta_{n-k+1}(m_1,m_2)$ is defined in (\ref{eq:telescopic2}), 
where the three terms in (\ref{eq:telescopic2}), correspond to the payoff obtained by transmitting a unit energy from node $1$, node $2$, and not transmitting any energy from any node, respectively. 
During time slot $n-k,$ if the available energy at the two nodes is $U^1_{n-k} =m_1, U^2_{n-k} = m_2, m_1, m_2 < k+1$, the optimal $F_{t}^i = 1, i=1,2$ if $q=i$, where $q$ is the
index of the maximum in  (\ref{eq:telescopic2}), and $F_{t}^i = 0, i=1,2$ if $q=3$. Finally, we can recursively compute the coefficients $\gamma_{n-k}(m_1, m_2)$ by taking expectation of  (\ref{eq:payoffmultisource}).
\end{proof}
\subsection{Sub-optimal policy for multiple nodes}\label{sec:subopt}
The energy transmitted by the optimal policy (section \ref{sec:multiplesources}) at any slot depends on thresholds $\gamma(i,j)$'s, where $i,j$ is the energy available at nodes $1$ and $2$, respectively. Thus, $B^2$ values of $\gamma(i,j)$'s have to be computed for each slot. If $B$ is large, this becomes quite significant.

A  simpler alternative  is a sub-optimal policy $\mathcal S$ that makes decoupled decisions at the two nodes. Thus, only $2B$ thresholds $\gamma$'s, $B$ for each node, have to be computed. Let us discuss the details for the binary case for ease of exposition.

Fix a given sample path for the energy arrivals to the two nodes, and for the channel gain realizations. Assume that each node ignores the other node, and independently computes the optimal policy described  in Section \ref{sec:singlesource}. Denote the decoupled decisions of node $i$ from Section \ref{sec:singlesource} at slot $k$
as $F_k(i) \in\{0,1\}$.  Our sub-optimal policy $\mathcal S$ operates as follows. During slot $k$, if $F_k(1)=F_k(2)=1$, transmit  from unit energy from node with larger channel gain, i.e., from node $i^{\star}$, where $i^{\star}=\arg\max_{i=1,2}h_k^i$. If $F_k(1)=1, F_k(2)=0$, transmit unit energy from node $1,$ and vice-versa. Finally, if  $F_k(1)=F_k(2)=0$,  do not transmit from either node.
Let us call the  optimal policy of Section \ref{sec:multiplesources} as $\mathcal O$.
\begin{figure}
\centering
\includegraphics[width=3.25in]{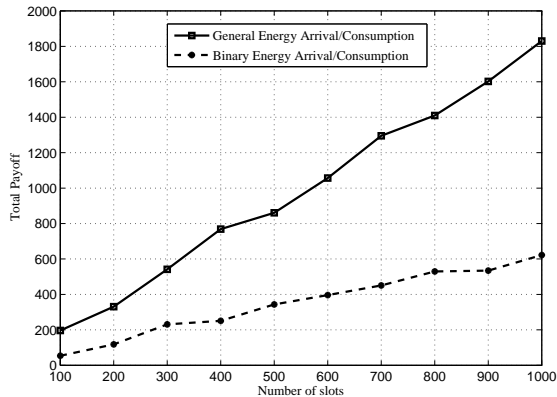}
%\vspace{-0.2in}
\caption{Performance of the single node optimal policy for $B=10$. The solid curve represents optimal payoff under uniform energy arrivals in [0:10], and the dashed curve corresponds to binary energy arrivals and transmission.}
\label{fig:singlesource}
\end{figure}

\begin{figure}
\centering
\includegraphics[width=3.5in]{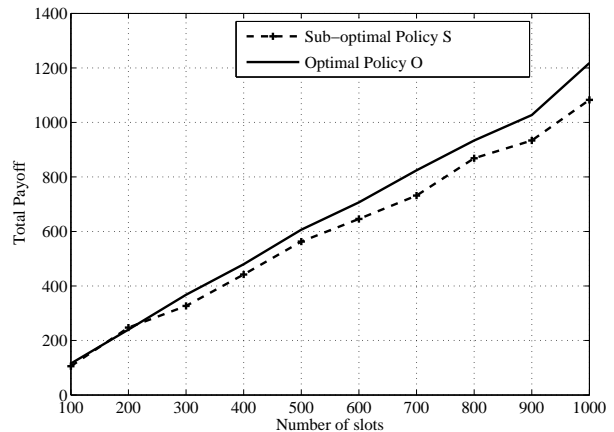}
%\vspace{-0.2in}
\caption{Performance comparison of optimal policy $\mathcal O$ and sub-optimal policy $\mathcal S$ for two nodes.}
\label{fig:twosources}
\end{figure}

\begin{lemma}\label{lem:subopt}
$\mathcal S$ and $\mathcal O$ differ only when $F_k(1)=F_k(2)=0$.
\end{lemma}
\begin{proof}
Claim: If $F_k(1)=F_k(2)=1$, then $\mathcal O$ also transmits from the node with the higher current channel gain. Proof by contradiction. Without loss of generality, at slot $k$, let $h_k^1 > h_k^2$.
Consider the two cases: i) $\mathcal O$ does not transmit any energy from any node at slot $k$. Then $\mathcal O$ is conserving energy at node $1$ for later use. Let it use that energy from node $1$ at slot $\ell> k$ (if it does not use that energy at all, then its wasted). But since $F_k(1)=1$, the optimal policy for just the single node $1$, there is no gain in shifting energy from slot $k$ to $\ell$ at node $1$, so $\mathcal O$ will also transmit from node $1$ at slot $k$.
 ii) $\mathcal O$ transmits unit energy from node $2$ that has lower channel gain. The essential idea remains the same that as above that if $\mathcal O$ transmits from node $2$, it is saving energy at node $1$ for later use, but since $F_k(1)=1$, it is better to use node $1$ now rather than later. Moreover, since $h_k^1 > h_k^2$, $\mathcal O$ will gain more by transmitting unit energy from node $1$ than node $2$ in slot $k$.

Otherwise, if $F_k(1)=1, F_k(2)=0$, then $\mathcal O$ transmits unit energy from node $1$ similar to $\mathcal S$. Argument is exactly as above.
\end{proof}

Lemma \ref{lem:subopt} shows that policies $\mathcal S$ and $\mathcal O$ could potentially differ only in slots where neither node transmits under $\mathcal S.$ Typically, such a scenario happens infrequently and thus the payoff obtained by $\mathcal S$ and $\mathcal O$ is expected to be similar (see Fig. \ref{fig:twosources}). The sub-optimal policy $\mathcal S$ can be easily extended for more than $2$ nodes and non-binary transmission, although obtaining analytical results seems difficult.

\section{Numerical Results} Fig.~\ref{fig:singlesource} plots the performance of the optimal policy for the single node case, with battery size $B=10$. The solid curve represents optimal payoff under uniform energy arrivals in [0:10], and the dashed curve corresponds to Bernoulli energy arrivals at rate $0.5,$ and binary transmission.

In Fig. \ref{fig:twosources}, we plot the optimal policy for the two nodes case with Bernoulli energy arrivals (at rate $0.5$ for each node), and binary transmission. In Fig. \ref{fig:twosources}, we also plot the performance of the sub-optimal policy $\mathcal S$ (Section \ref{sec:subopt}), and see that the performance of $\mathcal S$ is very close to the optimal policy as suggested by Lemma \ref{lem:subopt}.

\section{Concluding Remarks}
We presented exact optimal policies for maximizing utility over finite horizon in an energy harvesting system, assuming that energy arrival and expenditure are both discrete valued. Typically, finding closed form optimal policies is a hard problem. However, restricting ourselves to a discrete energy model allowed us to explicitly characterize the optimal policy. Indeed, we were able to express the optimal policy as a threshold policy where thresholds can be pre-determined using a recursive relation. Our methods could be applicable to other related problems as well.

\bibliographystyle{IEEEtran}
\bibliography{references}

\end{document}